\documentclass[a4paper,twoside,12pt,english]{article}

\usepackage[utf8]{inputenc}

\usepackage[english]{babel}
\usepackage{amsmath}
\usepackage{amsthm}
\usepackage{float}
\usepackage{caption} 
\usepackage[pdftex]{graphics}
\usepackage{tikz}

\usepackage{natbib}
\usepackage[colorlinks,linkcolor=blue,citecolor=blue,urlcolor=blue,breaklinks]{hyperref}

\makeatletter

\newcommand\dateymd{\number\year, \ifcase\month\or
January\or February\or March\or April\or May\or June\or
July\or August\or September\or October\or November\or
December\fi, \number\day}
\newcommand\printtime{%
\c@hours=\time \divide\c@hours by60
\c@minutes=\c@hours \multiply\c@minutes by-60
\advance \c@minutes by \time
\ifnum\c@hours<10 0\fi\the\c@hours:%
\ifnum\c@minutes<10 0\fi\the\c@minutes}
\makeatother

\usepackage[OT1]{fontenc}

\usepackage{fancyhdr}
\usepackage{setspace}

\fancypagestyle{footonly}{
\fancyhf{}

}

\theoremstyle{plain}
\newtheorem{proposition}{Proposition}
\newtheorem{lemma}{Lemma}
\newtheorem{theorem}{Theorem}[section]
\newtheorem{corollary}{Corollary}

\begin{document}

\thispagestyle{empty}
\null\vskip-20mm\null
\begin{center}
\hrule
\vskip7.5mm
\large
\textbf{\uppercase{The method of Enestr\"om and Phragm\'en}}\par
\vskip2pt
\textbf{\uppercase{for parliamentary elections}}\par 
\vskip2pt
\textbf{\uppercase{by means of approval voting}}\par
\vskip4pt
\textsc{Rosa Camps,\, Xavier Mora \textup{and} Laia Saumell}
\par
Departament de Matem\`{a}tiques,\break
Universitat Aut\`onoma de Barcelona,\break
Catalonia
\par\medskip
\texttt{xmora\,@\,mat.uab.cat}
\par\medskip
May 28, 2019; \,revised October 4, 2023
\vskip7.5mm
\hrule
\end{center}

\null\vskip-7.5mm\null

\vfil
\onehalfspacing

\begin{abstract}
We study a method for proportional representation that
was proposed at the turn from the nineteenth to the twentieth century by
Gustaf Eneström and Edvard Phragmén.
Like Phragmén's better-known iterative minimax method,
it is assumed that the voters express themselves by means of approval voting.
In contrast to the iterative minimax method, however,
here one starts by fixing a quota,
i.\,e.~the number of votes that give the right to a seat.
As a matter of fact, the method of Eneström and Phragmén can be seen as
an extension of the method of largest remainders from closed lists to open lists,
or also as
an adaptation of the single transferable vote to approval rather than preferential voting.
The properties of this method are studied and compared with those of other methods of the same kind.
\end{abstract}
\vskip3mm

\singlespacing
\vfil
\newpage

\newcommand\apr{\,\surd\,}
\newcommand\notapr{\kern-.1em\not\kern-.45em\apr}
\newcommand\notap{\,\not\kern-0.15em\apr}

\newcommand\ccands{I}
\newcommand\ncands{\ell}
\newcommand\tipus{k}
\newcommand\nvots{v}
\newcommand\nescons{n}		
\newcommand\quota{q}
\newcommand\quotaa{q'}
\newcommand\nvotstipus{v}
\newcommand\nvotscand{w}
\newcommand\capacitat{m}
\newcommand\scands{J}
\newcommand\nvotsexact{v}
\newcommand\nvotsalgun{w}
\newcommand\nvotstots{y}
\newcommand\nvotsnomes{z}
\newcommand\iesco{s}		
\newcommand\factor{r}

\newcommand\pas[1]{\unskip[#1]}
\newcommand\ini{\unskip[0]}
\newcommand\ara{\unskip[s]}
\newcommand\sseg{s\kern-.05em+\kern-.05em1}
\newcommand\seg{\unskip[\sseg]}
\newcommand\sant{s\kern-.05em-\kern-.05em1}
\newcommand\ant{\unskip[\sant]}
\newcommand\npen{n\kern-.05em-\kern-.05em1}
\newcommand\pen{\unskip[\npen]}
\newcommand\fin{\unskip[n]}
\newcommand\nseg{n\kern-.05em+\kern-.05em1}
\newcommand\atp{\unskip[p]}
\newcommand\att{\unskip[t]}

\newcommand\nvotscandmax{w_*}
\newcommand\imax{i_*}
\newcommand\droop{q_D}
\newcommand\mult{p}

\newcommand\pq[2]{\raise.25ex\hbox{\footnotesize${#1}\over{#2}$}%
\hskip-.35ex\null}
\newcommand\pqbis[2]{\leavevmode\raise.5ex\hbox{\scriptsize$#1$}\hskip-.2ex\raise.25ex\hbox{\scriptsize$/$}\hskip-.2ex\lower.2ex\hbox{\scriptsize$#2$}}
\newcommand\onehalf{\frac12}

\newcommand\itemind[1]{\leavevmode\hbox to15mm{\hss$#1$\hss}\ignorespaces}

%
%

\section*{Introduction}

\vskip-0.5mm
Recently, attention has been drawn to 
methods of proportional representation
where the electors express themselves by means of approval voting,
i.\,e.~each elector indicates an unordered list of the candidates that he approves of to represent him.
This has lead to a renewed interest in several methods of this kind
that had been devised at the end of the 19th century.
See, for instance, \cite{janson:2016, janson:2018}, \,\cite{brams:2019}, \,\cite{brill:2018} and the references therein. 



\smallskip
This note looks at a
method of this kind
that is associated with the names of Gustaf Eneström and Edvard Phragmén.
As in other better-known works of Phragmén (\citeyear{phragmen:1894}, \citeyear{phragmen:1895}, \citeyear{phragmen:1896}, \citeyear{phragmen:1899}), it is assumed that the electors express themselves by means of approval voting.
However, the procedure differs from that of those other works
in that here one starts by fixing a quota, i.\,e.~the number of votes that give the right to a seat.

More specifically, this method can be seen as an extension
of the method of largest remainders from party lists to (unordered) open lists.
On the other hand, it can also be seen as an adaptation of the single transferable vote
to approval voting instead of preferential voting.



\smallskip
This method was clearly formulated by Eneström in \citeyear{enestrom:1896}.
After that, Gustav Cassel described the same procedure in \citeyear{cassel:1903} 
under the name of ``Phragmén's first method''.
In fact, the main principle of the procedure can be found 
in a report of a lecture that Phragmén had given in \citeyear{phragmen:1893}\linebreak
(although the final proposal of that lecture was based on preferential voting).
In the subsequent years, Phragmén concentrated on a different idea,
without any comment on Eneström's (\citeyear{enestrom:1896}) work
(which he surely knew about; see the newspaper letters \citealt{enestrom:1896a} and \citealt{phragmen:1896a},
published together on the same day).
Later on, however, he detailed several variations of an approval-voting procedure based on that principle.
This was done 
in three narrowly-circulated memoirs from 1906
(\citealt{phragmen:1906a}, \citeyear{phragmen:1906b}, \citeyear{phragmen:1906c}).

These memoirs were motivated by a request from Finland, where a parliamentary reform was on its way.
That request had been addressed to Gösta Mittag-Leffler and his colleague mathematicians at Stockholm,
who were also motivated by the case of Sweden itself \citep[p.\,512--513]{stubhaug:2010}.
Neither Finland nor Sweden adopted Phragmén's proposal.
In~fact, Finland adopted a procedure based upon preferential voting and the so-called harmonic Borda method
(which name we take from \citealt{janson:2018}).
And three years later Sweden adopted D'Hondt's rule
together with some supplementary rules for determining which particular candidates are picked from every party,
one of these rules being the so-called Thiele's addition method \citep{thiele:1895}.
More historical details can be found in \cite{janson:2016}.


\smallskip
As we shall see, the method of Eneström and Phragmén
---more specifically, the particular variation that will be singled out below---
enjoys certain proportionality properties which are not fulfilled by
the harmonic Borda method nor by Thiele's addition one.
On the other hand, the comparison with Phrag\-mén's iterative minimax method
---that of 1894--99, called ``Phragmén's second method'' by \cite{cassel:1903}---
is more closely contested, with some advantage in favour of Eneström and Phragmén
in what respects simplicity.

\smallskip
Besides individual candidates, we will also allow for a candidate to consist of several individuals. 
This is appropriate to deal with the hypothetical but meaningful case where
the options that are submitted to approval voting are disjoint party lists rather than individual candidates
(similarly to \citealt{hill:2011}, who advocates for preferential voting in terms of parties).
In this case, the method is expected to answer the question of how many seats should be allocated to each party.

\section{The procedure}



\subsection{Generalities}

The Eneström-Phragmén procedure distributes a given number of seats \linebreak among several candidates
which have been subject to approval voting by a set of electors.
It has an iterative character. At each step a seat is allocated to the candidate with the largest number of votes
and the used votes lose a fraction of their value in accordance with a previously fixed quota.

\bigskip\noindent
We will use the following notation:

\medskip\noindent
\itemind{n} 	number of seats to be distributed

\medskip\noindent
\itemind{I}		set of eligible candidates\\
\itemind{i}		a candidate\\
\itemind{m_i}	capacity of candidate $i,$ i.\,e.\ number of seats that it can fill\\
\itemind{}		\llap($=1$ for individual candidates) 

\medskip\noindent
\itemind{k} 		a `type' of votes (or electors)\\
\itemind{A_k}		set of candidates approved by the electors~$k$ \ (we assume $A_k\neq\emptyset$)\\
\itemind{k\!\apr i} another way to say that $i\in A_k$

\medskip\noindent
\itemind{v}		total number of votes\\
\itemind{v_k}	number of votes of type $k$ \ ($\sum_k v_k = v$)\\
\itemind{w_i}	number of votes that approve $i$: \ 
							$w_i \!=\! \sum_{k\apr i} v_k$\\
\itemind{n_i}	number of seats allocated to candidate $i$ \ ($n_i\le m_i$)

\medskip\noindent
\itemind{J} 	a subset of candidates\\
\itemind{m_J}	capacity of the set $J$: \ $m_J = \sum_{i\in J}m_i$\\
\itemind{v_J}	number of votes that approve exactly the set $J$: \
						\rlap{$v_J = \sum_{k,\,A_k=J} v_k$}\\
\itemind{y_J}	number of votes that approve all candidates from $J$: \
						$y_J = \sum_{k,\,A_k\supseteq J} v_k$\\
\itemind{}				obviously, \ $v_J \le y_J$\\
\itemind{n_J}	total number of seats allocated to members of  $J$: \ 
						$n_J = \sum_{i\in J} n_i$

\medskip\noindent
\itemind{q}		quota, i.\,e.\ number of votes that give the right to a seat

\medskip\noindent
\itemind{s}		ordinal number in the seat-allocation iterative procedure\\
\itemind{X\ara}	value of $X$ after allocating the $s$-th seat \ ($X = v,\,v_k,\,w_i,\,n_i,\,I,\dots$)\\
\itemind{I\ara}	set of candidates that are still eligible after allocating the $s$-th seat,
\itemind{}		i.\,e. such that $n_i\ara < m_i$

\subsection{Basic version}

In the works of Eneström and Phragmén one finds several variations in different respects.
In order to ensure better properties he have been led to
choose the particular combination that is described next
(which is not exactly any of the versions considered by those authors).


\medskip
To start with, one adopts the (unrounded) quota of Droop and Hagenbach-Bischoff
\citep{pukelsheim:2017}: 
\begin{equation}
\label{eq:droop}
q \,=\, v \,/\, (\nseg)
\end{equation}

\medskip
The $n$ seats are allocated by means of an iterative procedure.
In the following, $s$ is a counter that will increase from $0$ to $n.$
We start with $s=0,$
$v_k\ini = v_k$ (the number of votes of each type $k$),
$n_i\ini = 0$ (no seats have been allocated)
and $I\ini = I$ (all candidates are eligible).

\medskip
For each $s$ and every $i\in I\ara$ (the eligible candidates)
one considers the presently existing votes in support of candidate $i$:
\begin{equation}
\label{eq:wi}
w_i\ara \,=\, \sum_{k\apr i} v_k\ara.
\end{equation}
Seat $\sseg$ is allocated to a candidate $\imax\in I\ara$ with a maximum support, 
i.\,e.\ such that
\begin{equation}
\label{eq:maximizer}
w_{\imax}\ara \,=\, \max_{i\in I\ara}\,w_i\ara \,=:\, w_*\ara
\end{equation}
(if several maximizers exist, any of them is allowed). So
\begin{equation}
\label{eq:allocation}
n_{\imax}\seg \,=\, n_{\imax}\ara + 1,
\end{equation}
and this candidate ceases to be eligible if its capacity has been reached:
\begin{equation}
I\seg \,=\, I\ara\setminus\{\imax\},\quad \text{if $n_{\imax}\seg = m_{\imax}$}
\end{equation}

%
%
%

\medskip
The allocation of seat~$\sseg$ to candidate $\imax$ is done
at the expense of
a certain fraction of the $w_*\ara$ votes that supported $\imax.$
More specifically, the votes in support of $\imax$ are considered to lose value in accordance with a common factor:
\begin{alignat}{2}
\label{eq:reduction}
v_k\seg \,&=\, \left(1 - \frac{q}{w_*\ara}\right)\,v_k\ara,&\qquad &\text{for $k\apr\imax$ and $w_*\ara\ge q,$}\\
\label{eq:exhausted}
v_k\seg \,&=\, 0,&\qquad &\text{for $k\apr\imax$ and $w_*\ara< q,$}\\
\label{eq:noreduction}
v_k\seg \,&=\, v_k\ara,&\qquad &\text{for $k\kern-.1em\not\kern-.45em\apr\imax.$}
\end{alignat}


\medskip
As one can easily check, in the case $w_*\ara\ge q$ the following equalities hold:
\begin{align}
\label{eq:exactupdate}
w_{\imax}\seg \,&=\, w_{\imax}\ara - q,\\
\label{eq:exactupdateglobal}
v\seg \,&=\, v\ara - q.
\end{align}
That is, the seat~$\sseg$ has been allocated at the expense of exactly one quota.

\medskip
In the case $w_*\ara<q$ the seat is allocated at the expense of all the existing votes in support of $\imax$,
even though they do not complete a whole quota. So in this case one has
\begin{align}
\label{eq:exhaustedtotal}
w_{\imax}\seg \,&=\,0,\\
\label{eq:decrementlessthanquota}
v\seg \,&>\, v\ara - q.
\end{align}


\medskip
If $\sseg = n,$ we are finished.
Otherwise, the procedure is repeated with $s$ replaced by $\sseg.$

\bigskip\noindent
\textit{Remark 1.1}. We are assuming that the set of candidates $i$ such that $w_i > 0$ contains at least $n$ elements.

\bigskip\noindent
\textit{Remark 1.2}. From \eqref{eq:reduction}--\eqref{eq:noreduction} it is clear that $v_k\ara,$
$w_i\ara,$ and $w_*\ara$ are non-increasing functions of $s.$



\newcommand\xxsep{\hskip.7em}
\newcommand\xsep{\hskip.65em}

\medskip\bigskip\noindent
\textit{Example 1.1}. \label{exemple11}
Consider the election of 3~representatives out of 9~candidates $a,b,e,f,u,v,x,y,z$
with the following approval votes:
$$
\hbox to\hsize{\xxsep$s=0:\xxsep 21\ a\,b\,x,\xsep 20\ a\,b\,e\,f,\xsep 19\ e\,f\,u\,v,\xsep 13\ u\,v,\xsep 10\ x\,y,\xsep 15\ z,\xsep 2\ a\,e\,u.$\hfill}
$$
The total number of votes is $v=100.$ Therefore, the quota is $q = v/(\nseg) = 100/4 = 25.$
From the preceding votes, the approval support $w_i$ for each candidate $i$ is found to be as follows:
$$
\hbox to\hsize{\xxsep$s=0:\xxsep a\ 43,\xsep b\ 41,\xsep e\ 41,\xsep f\ 39,\xsep u\ 34,\xsep v\ 32,\xsep x\ 31,\xsep y\ 10,\xsep z\ 15.$\hfill}
$$
The highest value is that of candidate $a,$ which therefore becomes elected. In accordance with \eqref{eq:reduction}, the votes that contain $a$ get their value reduced by the factor $(1-25/43)=0.419$. This results in the following figures:
$$
\hbox to\hsize{\xxsep$s=1:$\xxsep\scalebox{0.9}{$8.791\ a\,b\,x,\xsep 8.372\ a\,b\,e\,f,\xsep 19\ e\,f\,u\,v,\xsep 13\ u\,v,\xsep 10\ x\,y,\xsep 15\ z,\xsep 0.837\ a\,e\,u.$}\hfill}
$$
The corresponding approval supports are now as follows, where we use parentheses to indicate those candidates that have already been elected and therefore are not eligible any more:
$$
\hbox to\hsize{\xxsep$s=1:$\xxsep\scalebox{0.8}{$(a\ 18),\xsep b\ 17.163,\xsep e\ 28.209,\xsep f\ 27.372,\xsep u\ 32.837,\xsep v\ 32,\xsep x\ 18.791,\xsep y\ 10,\xsep z\ 15.$}\hfill}
$$
So the second elected candidate is $u$. The votes that contain this candidate are now reduced by the factor $(1-25/32.837)=0.239$:
$$
\hbox to\hsize{\xxsep$s=2:$\xxsep\scalebox{0.8}{$8.791\ a\,b\,x,\xsep 8.372\ a\,b\,e\,f,\xsep 4.535\ e\,f\,u\,v,\xsep 3.103\ u\,v,\xsep 10\ x\,y,\xsep 15\ z,\xsep 0.200\ a\,e\,u,$}\hfill}
$$
which results in the following approval supports:
$$
\hbox to\hsize{\xxsep$s=2:$\xxsep\scalebox{0.72}{$(a\ 17.363),\xsep b\ 17.163,\xsep e\ 13.107,\xsep f\ 12.907,\xsep (u\ 7.837),\xsep v\ 7.637,\xsep x\ 18.791,\xsep y\ 10,\xsep z\ 15.$}\hfill}
$$
So $x$ is elected in third place, in spite of having less than a quota. 
The three elected candidates are thus $a,\,u$ and $x.$
Besides that, later on (\S\,\ref{substituts}) we will need to know also the remaining votes, which are as follows:
$$
\hbox to\hsize{\xxsep$s=3:$\xxsep\scalebox{0.8}{$0\ a\,b\,x,\xsep 8.372\ a\,b\,e\,f,\xsep 4.535\ e\,f\,u\,v,\xsep 3.103\ u\,v,\xsep 0\ x\,y,\xsep 15\ z,\xsep 0.200\ a\,e\,u,$}\hfill}
\label{romanents}
$$


\subsection{Variations}

\medskip
\paragraph{Simple fractions.} 
\label{pg:simple_fractions}
In \citeauthor{phragmen:1906a} (\citeyear{phragmen:1906a}, \citeyear{phragmen:1906b}, \citeyear{phragmen:1906c})
equations \eqref{eq:reduction}--\eqref{eq:exhausted} are replaced by the following one:
\begin{equation}
\label{eq:roundedreduction}
v_k\seg \,=\, \left(1 - \frac{1}{\lceil w_*\ara\,/\,q\rceil}\right)\,v_k\ara,\qquad\text{for $k\apr\imax.$}\\ 
\end{equation}
Notice that this equation agrees with \eqref{eq:reduction}--\eqref{eq:exhausted} in the case $w_*\ara\le q.$
However, for non-integer $w_*\ara/q > 1$ the seat is allocated in exchange of less than a whole quota;
in other words, the equality signs of \eqref{eq:exactupdate} and \eqref{eq:exactupdateglobal} are replaced by ``$\ge$''.
The rational factor in \eqref{eq:roundedreduction} was probably intended towards facilitating the computation by hand.
Another possible motivation was the comparison with Thiele's addition method (see \S\ref{ssec:thiele-addition}).



\paragraph{Other quotas.} 
The initial versions of this method
(\citealt{enestrom:1896}, \citealt[p.\,47--50]{cassel:1903}, \citealt{phragmen:1906a}, \citeyear{phragmen:1906b})
made use of the simple (a.\,k.\,a.\ Hare) quota $q = v \,/\, n.$
With this quota, the above procedure is an extension of the standard method of largest remainders.
According to \citet{mittagleffler:1906b}, the quota of Droop and Hagenbach-Bischoff achieves better results in comparison with Thiele's method
(which is related to the bad threshold of the latter, see \citealt{janson:2018}). 
In this connection, it should also be noted that the quota of Droop and Hagenbach-Bischoff had already been used by \cite{phragmen:1893} in computations of the kind of \eqref{eq:reduction}--\eqref{eq:noreduction}.


\paragraph{Quota updated at each step.}
In \citet{phragmen:1906a} a formula like \eqref{eq:roundedreduction} was used,
but instead of the original (Hare) quota $q = v\ini / n,$
an updated one was used, namely $q'\ara = v\ara / (n - s).$

\paragraph{Votes cease to count when they become empty}
(this variation makes sense only together with the preceding one). 
In this variant, equation \eqref{eq:reduction} is used only when the votes of type $k$ contain further eligible candidates,
i.\,e.\ when $A_k\cap I\seg\neq\emptyset;$ otherwise one puts $v_k\seg = 0.$

%
%

\paragraph{Threshold condition.}
If $w_*\ara$ is much smaller than $q$
(which case can easily arise for a large number of candidates)
then $\imax$ gets a seat in spite of having much less than a quota.
In this connection \citeauthor{phragmen:1906b} (\citeyear{phragmen:1906b}, \citeyear{phragmen:1906c}) proposed to require
a condition of the form $w_*\ara / q \ge \alpha$
for some previously fixed $\alpha\in[0,1];$ if this condition is not satisfied,
then the procedure would be stopped and a new election would be called.
More specifically he suggested $\alpha = \pqbis34$ (\citeyear{phragmen:1906b}, \citeyear{phragmen:1906c})
and in (\citeyear{phragmen:1906a}) he had taken $\alpha = \pqbis12.$
In our basic version we did not include such a condition,
which amounts to $\alpha = 0$
(and corresponds to the usual formulation of the method of largest remainders).


\paragraph{Negative numbers of votes.}
\label{pg:negative_numbers}%
This variant would apply \eqref{eq:reduction} no matter whether $w_*\ara$ is larger or smaller than $q.$
As we have already seen, this could result in negative values of $v_k\seg$ for $k\apr\imax.$
However, this could still make sense 
for the underlying overrepresentation to be corrected in subsequent steps.

\subsection{Uninominal voting}
\label{sec:uninominal}

The following propositions specify the behaviour of the Eneström-Phragmén method in the case of uninominal voting,
i.\,e.~when every elector approves one candidate and only one.
The first proposition, whose proof is obvious, is about the case of individual candidates.
The second one is about the case of disjoint party lists.

\begin{proposition}
Assume that all candidates have capacity one
and that every elector approves one of them and only one.
In this case the Eneström-Phragmén method amounts to selecting the $n$ most voted candidates.
\end{proposition}

In the case of individual candidates,
the uninominal situation is certainly quite undesirable in the spirit of proportional representation.
Because strategies become possible whereby a set of electors
can get more representatives than they have a right to.
Nonetheless, in such a situation there is no better way than selecting the most voted candidates.
Hopefully, in practice electors will be lead to approve more than one candidate.
Those with a minoritary opinion will do it so as to provide more likely alternative candidates.
Those with a majoritary opinion will do it so as to obtain more representatives.

\smallskip
In the case of disjoint party lists
the uninominal situation still allows for a result in the spirit of proportional representation:

\begin{proposition}
\label{st:uninom-party}
Assume that all candidates have unlimited capacity
and that every elector approves one of them and only one.
In this case the Eneström-Phragmén method amounts to the method of largest remainders with the Droop quota.
\end{proposition} 

\begin{proof}
In the present situation each type~$k$ corresponds to a single candidate~$i$ and viceversa.
So we can write $v_i$ instead of $v_k$ or $w_i$. We will distinguish two cases.

\medskip
\noindent
Case~(a): $w_*\fin < q.$\quad
Let us consider the number 
\begin{equation}
\label{eq:deft}
t \,=\, \min\,\{s\mid w_*\ara < q\},
\end{equation}
and let $t_i$ mean the number of seats that are allocated to $i$ in the first $t$ steps of the procedure,
i.\,e.~for $s = 0\dots t-1.$
The hypothesis that defines this case ensures that $t \le n.$
By the definition of $t,$ for $s<t$ every seat allocation is done at the expense of exactly one quota.
Therefore $v_i\att = v_i\ini - t_i\,q,$ or equivalently,
\begin{equation}
\label{eq:divent}
v_i\ini = t_i\,q + v_i\att,\qquad
\text{ where $0 \le v_i\att < q.$ }
\end{equation}
Here, the strict inequality at the right holds because $v_i\att \le w_*\att < q.$
So, $v_i\att$ is the remainder of dividing the number of votes $v_i\ini$ by the quota $q.$
If $t = n,$ there are no more seats to allocate and $n_i = t_i.$
It $t < n,$ then for $s$ in the interval $t\le s<n$ the corresponding seat is allocated in exchange of the whole remainder of votes $v_i\att$
(which can happen only once for each candidate).
So these $n-t$ seats will be allocated to the $n-t$ candidates with greatest remainders $v_i\att.$
\newcommand\ttt{\tau}

\medskip
\noindent
Case~(b): $w_*\fin \ge q.$\quad
In this case we can write
$$q \,\le\, w_*\fin \,\le\, v\fin \,=\, v\ini - n q \,=\, q,$$
where we have used the facts that $w_*\fin = \max_i v_i\fin$ and $v\fin = \sum_i v_i\fin.$
So, all the terms that we have just displayed are equal to each other.
In particular, $w_*\fin = q = v\fin.$
Now, this implies the existence of some candidate $j$ such that $v_j\fin = q$ whereas $v_i\fin = 0$ for $i \neq j.$
So $j$ is getting one seat less than the number of quotas contained in $v_j\ini.$
This can only happen for all the $v_i\ini$ being exact multiples of $q,$
and therefore all the positive $v_i\ara$ being equal to $q$ from some step $s$ on.
So the seat $\sseg$ can be allocated to any of the candidates with $v_i\ara = q > 0$
and such ties carry on to all the remaining seats,
just in the same way as largest remainders (with the Droop quota).
\end{proof}

\section{Proportionality properties}

\medskip
According to
\cite{phragmen:1906b}, ``for the mathematically trained reader the proportional character of the proposed rule for voting power reduction should be clear without further explanation''.
In this section we establish a couple of proportionality properties of the kind that is associated with the so-called ``exclusion thresholds'' (\citealt{janson:2018}).

\medskip\noindent
Let us recall that, for any given set $J$ of candidates, \,$v_J$\, and \,$y_J$\, mean respectively the number of votes that approve exactly the set $J$ and the number of those that approve \emph{at least} the set $J.$ In this connection, we will use also the following notation:
\begin{equation}
\label{eq:jstar}
J^* = \bigcup\,\{A_k\mid k\text{ such that } A_k\supseteq J\}.
\end{equation}
So $i\in J^*$ if and only if $i$ is approved by at least one voter who approves (also) all the elements of $J.$

\begin{theorem}
\label{st:droop}
For any subset $J$ of candidates and any $\ell \le \min(n, m_J),$\linebreak
if \,$v_J > \ell\,q$\, then \,$n_J \ge \ell.$
\end{theorem}

\begin{theorem}
\label{st:pjr}
For any subset $J$ of candidates and any $\ell \le \min(n, m_J),$\linebreak
if \,$y_J > \ell\,q$\, then \,$n_{J^*} \ge \ell.$
\end{theorem}

\noindent
\textit{Preparation for the proofs.}
Instead of the number \eqref{eq:deft}, here we will consider the slightly different number
\begin{equation}
\label{eq:defp}
p \,=\, \min\,\{s\mid\nvotscandmax\ara\le\quota\}.
\end{equation}

\smallskip\noindent
We claim that $p\le \nescons,$ i.\,e.\ $\nvotscandmax\ara\le\quota$ for some $\iesco\le\nescons$.
In particular, $\nvotscandmax\fin\le\quota,$ that is, after the allocation of the last seat, no candidate has more than a quota.
In fact, if one has $\nvotscandmax\ara>\quota$ for all $\iesco\le\nescons-1$
then ---by \eqref{eq:exactupdateglobal}---
$\nvots\fin = \nvots - \nescons\quota = \quota$
and therefore $w_*\fin = w_{\imax}\fin \le v\fin = \quota.$

\medskip\noindent
For any subset $\scands$ of candidates we will consider the number $p_\scands$ of candidates of $\scands$ that are elected in the first $p$ steps of the procedure,
i.\,e.\ for $\iesco = 0\dots p-1.$ Obviously, $p_\scands \le \nescons_\scands$. So in order to prove Theorems~\ref{st:droop} and \ref{st:pjr} it will suffice to prove respectively the inequalities $p_\scands \ge \ell$ and $p_{\scands^*} \ge \ell$.

\medskip\noindent
In the following proofs consideration is limited to $\iesco = 0\dots p-1,$
for which we are ensured that $\nvotscandmax\ara>\quota.$
The candidate that is chosen in step~$s$ will be denoted $\imax\ara.$ 

\smallskip
\begin{proof}[Proof of Theorem~\ref{st:droop}]
Since $v_J = \sum_{A_k=J} v_k,$ when $\imax\ara\in J,$ equation \eqref{eq:reduction} entails that 
$$
v_J\seg \,=\, \left(1 - \frac{q}{w_*\ara}\right) v_J\ara \,\ge\, v_J\ara - q,
$$
where we used also the fact that $w_*\ara = w_{\imax\ara}\ara \ge v_J\ara.$ On the other hand, when $\imax\ara\not\in J,$
equation \eqref{eq:noreduction} entails that $v_J\seg \,=\, v_J\ara.$ Altogether we get the inequality
$$
v_J\pas{p} \,\ge\, v_J - p_J\,q.
$$
Assume now that the hypothesis $v_J > \ell\,q$ holds, as well as the negation of the conclusion, namely $p_J \le \ell - 1.$
Since we are assuming $\ell \le m_J,$ there exists $j \in J$ that remains eligible after step $p.$
This allows to write the following chain of inequalities, which contradicts the definition of $p$:
$$
w_*\atp \,\ge\,
w_j\atp \,\ge\,
v_J\atp \,\ge\,
v_J - p_J\,q \,\ge\,
v_J - (\ell-1)\,q \,>\,
q.
\qedhere
$$
\end{proof}

\smallskip
\begin{proof}[Proof of Theorem~\ref{st:pjr}]
Here we are considering $y_J = \sum_{A_k\supseteq J} v_k$
and the set $J^* = \bigcup_{A_k\supseteq J}A_k.$
An element of $J^*$ need not be contained in all of these sets $A_k$.
However, when $\imax\ara\in J^*$ equations \eqref{eq:reduction} and \eqref{eq:noreduction} still allow us to derive the inequality
\begin{align}
y_J\seg
\,&=\, \sum_{\substack{A_k\supseteq J\\k\apr\imax}}v_k\seg + \sum_{\substack{A_k\supseteq J\\k\notap\imax}}v_k\seg \nonumber\\
\,&=\, \left(1 - \frac{q}{w_*\ara}\right)\,\sum_{\substack{A_k\supseteq J\\k\apr\imax}}v_k\ara + \sum_{\substack{A_k\supseteq J\\k\notap\imax}}v_k\ara
\,\ge\, y_J\ara - q,
\end{align}
where we used the fact that $w_*\ara = w_{\imax}\ara \ge \sum_{\,k\apr\imax,\,A_k\supseteq J}v_k\ara.$
On the other hand, when $\imax\ara\not\in J^*,$ equation \eqref{eq:noreduction} ensures that $y_J\seg \,=\, y_J\ara.$
Therefore, we have
$$
y_J\pas{p} \,\ge\, y_J - p_{J^*}\,q.
$$
Analogously to the proof of Theorem~\ref{st:droop}, this inequality together with the hypothesis $v_J > \ell\,q$ and the negation of the conclusion, namely $p_{J^*} \le \ell - 1,$ contradicts the definition of $p$ because of the following chain of inequalities, where $j$ stands for any element of $J$ that remains eligible after step $p$ (such an element exists because of the hypothesis that $\ell \le m_J$):
$$
w_*\atp \,\ge\,
w_j\atp \,\ge\,
y_J\atp \,\ge\,
y_J - p_{J^*}\,q \,\ge\,
y_J - (\ell-1)\,q \,>\,
q.
\qedhere
$$
\end{proof}


\bigskip\noindent
\textit{Remark~2.1}. By using \eqref{eq:deft} instead of \eqref{eq:defp} one can deal similarly with the weaker hypothesis $v_J \ge \ell\,q,$ resp.~$y_J \ge \ell\,q,$ with the result that the inequality $n_J \ge \ell,$ resp.~$n_{J^*} \ge \ell,$ can fail only
in certain singular cases with ties that allow for several different allocations. Even then, some of these allocations do satisfy the equality $n_J = \ell,$ resp.~$n_{J^*} = \ell.$




\bigskip
\newcommand\resp[1]{\unskip\,\ \hbox{\upshape[\,resp.}\ \hbox{#1\upshape\,]}}
\begin{corollary}[Majority preservation]
\label{st:majority}%
If $n$ is odd \resp{even}, $m_J \ge (n+1)/2$ \resp{$m_J \ge n/2$} and $v_J > v/2$ \resp{$v_J \ge v/2$}, then $n_J > n/2$ \resp{$n_J \ge n/2$}.
Analogously happens with $y_J$ and $n_{J^*}$ instead of $v_J$ and $n_J.$
\end{corollary}
\begin{proof}
It suffices to note that the hypothesis $v_J > v/2$ \resp{$v_J \ge v/2$} implies $v_J > \ell\,q$ with $\ell = (n+1)/2 > n/2$ \resp{$\ell = n/2$}.
\end{proof}

\smallskip\noindent
\textit{Remark~2.2.} This property is not satisfied when the Hare quota is used instead of the Droop one. At the end of the 19-th century this led several swiss institutions to adopt the method of largest remainders with the quota of Droop and Hagenbach-Bischoff (see \citealt[p.\,171,\,197,\,276]{kloti:1901}, and \citealt[\S\,3.2]{kopfermann:1991}).

\smallskip\noindent
\textit{Remark~2.3.} In exchange for the property contained in the preceding corollary, one cannot avoid the possibility of having $v_J < v/2$ but $n_J > n/2.$
For instance, for $n=5$ the votes $56\,A,\ 34\,B,\ 30\,C$ result in the seats $3\,A,\ 1\,B,\ 1\,C,$
where $A$ has less than half the votes but more than half the seats.

\bigskip\smallskip
In the notation of \cite{janson:2018}, Theorems~\ref{st:droop} and \ref{st:pjr} ensure that both $\pi_{\hbox{\scriptsize\sffamily same}}(\ell,n)$ and $\pi_{\hbox{\scriptsize\sffamily PJR}}(\ell,n)$ are $\le \ell/(\nseg).$ \marginpar{$\ell\,q$} 
On the other hand, taking into account the case of two parties with respectively $\ell\,q$ and $(n+1-\ell)\,q$ votes,
one sees that both $\pi_{\hbox{\scriptsize\sffamily same}}(\ell,n)$ and $\pi_{\hbox{\scriptsize\sffamily PJR}}(\ell,n)$ are equal to the optimal value $\ell/(\nseg).$ \marginpar{$\ell\,q$} 

\bigskip\noindent
\textit{Remark~2.4.} For Thiele's global optimization method (see \S\ref{sec:thiele}) it has been shown that in the conditions of Theorem~\ref{st:pjr} one is ensured the following additional property \cite[Theorem 7.6]{janson:2018}: There exist some electors~$k$ such that $A_k \supseteq J$ and $A_k$ contains at least $\ell$ elected candidates. This property does not hold for Phragmén's iterative minimax method \cite[Example 7.4]{janson:2018}.
Neither does it hold for the method that we are studying here,
a counterexample being, for instance, the following:
$$21\,a\,b\,c_1,\ 21\,a\,b\,c_2,\ 22\,c_1\,c_2\,c_3,\ 1\,c_1\,c_3,\ 15\,c_3.\quad n=3.$$
The seats are allocated to $c_1, c_2$ and $c_3$, in this order.
The above-mentioned property is not satisfied by the set $J=\{a,b\},$ for which $y_J = 42 > 2 q$

\bigskip\noindent
\textit{Remark~2.5.} Theorem~\ref{st:droop} does not hold for variant~\ref{pg:simple_fractions}.
A counterexample is the following: $95\,A,\ 79\,B,\ 75\,C,\ 7\,AB,\ 56\,AC$ with $n=3.$
As one can check, party~$B$ gets no seat in spite of the fact that it has more than a quota.

\bigskip\noindent
\textit{Remark~2.6.} The following example, taken from \cite[p.\,14]{phragmen:1906b},
shows that Theorem~\ref{st:pjr} does not hold if the hypothesis $\ell \le m_J$ is replaced by $\ell \le m_{J^*}$:
$$
120\ a_1\,a_2\,a_3\,a_4\,a_5,\quad 86\ b_1\,b_2\,b_3\,b_4,\quad 24\ a_2\,a_3\,a_4\,a_5.
$$
Let $n=7.$ The quota is $q = 230/8 = 28.75.$ The seats get allocated to $a_2,a_3,a_4,b_1,a_5,b_2$ and $b_3$, in this order.
For $J=\{a_2,a_3,a_4,a_5\}$ one has $J^*=\{a_1,a_2,a_3,a_4,a_5\}$ and $y_J = 144 > 5\,q.$ However, $n_{J^*} = 4 < 5.$

\section{Different kinds of monotonicity,\\or the lack of it}
\label{sec:mono}

\subsection{Monotonicity for individual candidates}

Let us consider the case of individual candidates
and a variation of the votes whereby
one candidate gets more approvals than before,
the other candidates keeping exactly the same approvals as before.
In such a situation one certainly expects the following property:
if that candidate was allocated a seat before the variation,
he will also be allocated a seat after it.
We refer to this property as \textbf{monotonicity for individual candidates}.

\begin{theorem}
Monotonicity for individual candidates holds.
\end{theorem}

\begin{proof}
In the following we use a tilde to mean what happens after the votes have been varied. 
Let $i$ be the candidate that gets additional approvals.
So $\tilde w_i > w_i,$ whereas $\tilde w_j = w_j$ for any $j\neq i.$
Let us assume that for the original votes $i$ is elected when $s$ takes the value $t$.
Consider now the modified votes. We claim that either $i$ is elected for some $s<t$,
or it is elected for $s=t$. In other words, the assumption that $i$ is not elected for any $s\le t-1$
implies that it is elected for $s=t$. In fact, that assumption entails the following facts for $s\le t$, which are obtained by induction: $\tilde w_i\ara \ge w_i\ara$; $\tilde w_j\ara = w_j\ara$ for any $j\in I\ara\setminus\{i\}$; the elected candidates are the same as with the original votes. This gives the desired result since $\tilde w_i\pas{t} \ge w_i\pas{t} \ge w_j\pas{t} = \tilde w_j\pas{t}$ for any $j\in I\pas{t}\setminus\{i\}.$
\end{proof}

\bigskip\noindent
\textit{Remark 3.1.}
In the event of ties it can happen that both the original votes and the modified ones admit one allocation where $i$ is elected and another one it is not. In such a case, depending on which allocation is chosen before and after the modification it may give the wrong feeling of a lack of monotonicity for individual candidates.

\bigskip\noindent
\textit{Remark 3.2.}
Monotonicity for individual candidates does not hold for variant~\ref{pg:negative_numbers}. 
An example is provided by the following votes: $7\,a,\ 3\,b,\ 2\,c,\ 1\,d_i$ ($i=1..18$).
Consider the election of 2~representatives. The quota is $q=10.$
As one can check, variant~\ref{pg:negative_numbers} results in the election of $a$ and $b.$
However, if the $7\,a$ votes are replaced by $4\,a$ plus $3\,ab,$ then $a$ and $c$ are elected instead of $a$ and~$b.$

%
%

\subsection{Lack of monotonicity for party lists}

Let us consider now the analogous situation with party lists instead of individual candidates.
So the votes undergo a variation where one of these party lists gets additional approvals
whereas the other party lists keep exactly the same approvals as before.
In such a situation one would expect that party to keep at least the same number of seats.
However, it is not always so:

\bigskip
\noindent
\textit{Exemple 3.1.}
Consider the election of 3~representatives from 3~party lists $A,\,B,\,C,$
with the following approval votes:
$$5\ A,\xxsep 4\ B,\xxsep 6\ A\,C,\xxsep  4\ BC.$$
As one can check, the Eneström-Phragmén procedure results in the seats being allocated successively to $A,\,B$ and $A.$
Now, if one of the votes that presently approve only $B$ is changed to approve also $A,$
the seats are then allocated successively to $A, C$ and $B.$


\bigskip
The following example illustrates a related but slightly different phenomenon
where the additional approval happens in an elector that was previously abstaining
(so the number of votes, and therefore the quota, get changed):

\bigskip\noindent
\textit{Exemple 3.2}.
Consider again the election of 3~representatives from 3~party lists $A,\,B,\,C,$
the initial votes being now
$$5\ A,\ 3\ B,\ 3\ AB,\ 8\ AC,\ 7\ BC.$$
As one can check, the seats are again allocated successively to $A,\,B$ and $A.$
Let us now add one vote that approves only $A.$
After this modification, the seats are allocated successively to $A,\,C$ and $B.$

\bigskip
These phenomena can be justified by arguing that
the Eneström-Phragmén procedure aims only at
electing a good set of representatives
in the spirit of proportional representation,
that is, in being well distributed among the electors.
Since the votes express only approval,
the issue whether an elector is represented by one or another candidate
is irrelevant as long as both of them have the approval of that elector.



\subsection{Lack of house monotonicity}

On the other hand, the procedure is aimed at a particular total number of seats, whose value determines the quota.
So it is not a surprise to see that going from $n$ to $\nseg$ seats
does not always amount to simply adding one candidate.
For instance, example~1.1 with $n=3$ results in electing successively $a,\,u$ and $x.$
But the same votes with $n=4$ result in electing successively $a,\,u,\,b$ and $z.$

%

\section{Asymptotic behaviour as $n\rightarrow\infty$ in the case of two parties}

\medskip
In this section we assume that there are only two parties, $A$ and $B$
---so every elector approves either $A$ or $B$ or both of them---
and we ask ourselves about the asymptotic behaviour of $n_A/n$ and $n_B/n$ as $n\rightarrow\infty$
and its dependence on the number of votes of each kind, $v_A, v_B, v_{AB}$.
Since the electors that approve both $A$ and $B$ are indifferent between all proportions of seats between $A$ and $B,$
the ideal behaviour is to reproduce the proportion between $v_A$ and $v_B,$ that is to have $\lim_{n\rightarrow\infty}n_A/n = v_A/(v_A+v_B)$.

\smallskip
For Phragmén's iterative minimax method \cite{mora:2015} noticed that this desirable behaviour is far from being satisfied;
instead, the dependence of $\lim_{n\rightarrow\infty}n_A/n$ on $v_A, v_B, v_{AB}$ exhibits a ``devil's staircase'' character,
which phenomenon has been mathematically dissected by \cite{janson:2017}.
In this section we will see that the method of Eneström and Phragmén is better behaved.

\medskip
We will make use of the following notation:
\begin{align}
\alpha\ara = v_A\ara/v,\quad \beta\ara &= v_B\ara/v,\quad \zeta\ara = v_{AB}\ara/v;\\[2pt]
\rho = q/v &= 1/(\nseg).
\end{align}
For~$\alpha\ini=\beta\ini=0$ it is easily seen that the procedure allows for any seat distribution between $A$ and $B$.
For~$\alpha\ini=0<\beta\ini$ ---resp.\,$\beta\ini=0<\alpha\ini$--- it is also easily seen that all seats are given to $B$ ---resp.\,$A$---.
\textit{So in the following developments we will assume $\alpha\ini,\,\beta\ini>0.$}
\label{standingassumption}
Concerning $\zeta\ini$ sometimes we will deal separately with the cases $\zeta\ini>0$ and $\zeta\ini=0$
(the latter case has been considered in Section~\ref{sec:uninominal}).

Instead of $\alpha\ara,\,\alpha\seg,\,\alpha[s\kern-.05em+\kern-.05em2],$
next we will write simply  $\alpha,\,\alpha',\,\alpha''$,
and similarly with $\beta$ and $\zeta.$
As it is shown in next lemma, we are always in the case $w_*\ara > q$
and the seat allocations spend always a whole quota.
More specifically, for $s\le n-1$
the values of $(\alpha',\beta',\gamma')$ remain positive
and they are related to $(\alpha,\beta,\gamma)$ in the following way: 
\begin{equation}
\label{eq:cases}
\begin{cases}
\,\alpha' = \left(1 - \dfrac\rho{\alpha+\zeta}\right)\,\alpha,\quad \beta' = \beta,\quad
\zeta' = \left(1 - \dfrac\rho{\alpha+\zeta}\right)\,\zeta,
&\text{if $\alpha > \beta$;}\\[12pt]
\,\alpha' = \alpha,\quad \beta' = \left(1 - \dfrac\rho{\beta+\zeta}\right)\,\beta,\quad
\zeta' = \left(1 - \dfrac\rho{\beta+\zeta}\right)\,\zeta,
&\text{if $\alpha < \beta$.}
\end{cases}
\end{equation}
In the first of these two cases the seat at stake is allocated to $A$ and in the second case it is allocated to $B.$
For $\alpha = \beta$ one is allowed to choose between both possibilities.


\begin{lemma}
\label{st:lemma1}
Assume that $\zeta\ini>0.$
The following facts take place:
\begin{alignat}{2}
\label{eq:ac}
&\max\,(\alpha + \zeta,\,\beta + \zeta) \,>\, \rho&\qquad&\forall s\le n-1,\\[2pt]
\label{eq:eqs}
&(\alpha',\beta',\gamma') \text{ are related to } (\alpha,\beta,\gamma) \text{ by \eqref{eq:cases}}
&\qquad&\forall s\le n-1,\\[2pt]
\label{eq:abc}
&\alpha + \beta + \zeta \,=\, 1 - s\,\rho \,=\, (n+1-s)\,\rho,&\qquad&\forall s\le n,\\[2pt]
\label{eq:pos}
&\alpha, \beta, \zeta \,>\, 0,&\qquad&\forall s\le n.
\end{alignat}
\end{lemma}

\begin{proof}
For $s=0$ \eqref{eq:abc} holds because of the definition of $\alpha,\beta,\zeta$,
and \eqref{eq:pos} is a standing hypothesis.
Still for $s=0,$ \eqref{eq:ac} holds because the contrary inequality, namely having both $\alpha + \zeta\le\rho$ and $\beta + \zeta\le\rho$ would imply $\alpha + \beta + \zeta < \alpha + \beta + 2\,\zeta \le 2\rho$, i.\,e.~$(\nseg)\rho < 2\rho$, i.e. $n < 1$.

On the oher hand, for any $s,$ it is easily checked that \eqref{eq:ac} implies \eqref{eq:cases} and that the resulting $(\alpha',\beta',\gamma')$ satisfy \eqref{eq:abc} (with $s$ replaced by $s+1$) and \eqref{eq:pos}.

So it only remains proving that \eqref{eq:ac} gets reproduced as long as $s\le n-1.$
Let $s$ be the last time that \eqref{eq:ac} holds and assume that $s\le n-2.$
As we have just seen, this implies $\alpha' + \beta' + \zeta' = (n-s)\rho$ and $\alpha'\!,\,\beta'\!,\,\zeta' > 0.$
But we are assuming that $\max\,(\alpha' + \zeta',\,\beta' + \zeta') \le \rho.$
Without loss of generality, we can assume also that $\max\,(\alpha' + \zeta',\,\beta' + \zeta') = \alpha' + \zeta'.$
By combining these facts with the known value of
$\alpha' + \beta' + \zeta',$ we get $\beta' = (n-s)\rho - (\alpha' + \zeta') \ge (n-s-1)\rho \ge \rho,$
where we have used the assumption that $s\le n-2.$ On~account of the inequality $\zeta' > 0,$ it follows that $\beta' + \zeta' > \rho,$ in contradiction with the above assumption that $\max\,(\alpha' + \zeta',\,\beta' + \zeta') \le \rho.$
%
%
\end{proof}

\begin{lemma}
\label{st:lemma2}
Assume that $\zeta\ini>0.$
If seat $s+1$ is allocated to $A$ and seat $s+2$ is allocated to $B,$
then seat $s+3$ is necessarily allocated to $A$ (whenever $s+3\le n$).
Analogously happens for $A$ interchanged with $B.$ 
\end{lemma}

\begin{proof}
Giving the seat $s+1$ to $A$ and seat $s+2$ to $B$ implies
\begin{gather}
\alpha' = (1 - \frac\rho{\alpha+\zeta})\,\alpha,\quad \beta' = \beta,\quad \zeta' = (1 - \frac\rho{\alpha+\zeta})\,\zeta;
\\[5pt]
\alpha'' = \alpha',\quad \beta'' = (1 - \frac\rho{\beta'+\zeta'})\,\beta',\quad \zeta'' = (1 - \frac\rho{\beta'+\zeta'})\,\zeta'.
\end{gather}
Besides, we are ensured that $\alpha/\beta \ge 1$ (and $\alpha'/\beta' \le 1$).
Our claim will be proved if we show that $\alpha''/\beta'' > \alpha/\beta$
(which implies $\alpha''/\beta'' > 1).$
Since
$$
\frac{\alpha''}{\beta''} \,=\, \frac{1}{(1-\rho/(\beta'+\zeta'))}\,\frac{\alpha'}{\beta'} \,=\, \frac{(1-\rho/(\alpha+\zeta))}{(1-\rho/(\beta'+\zeta'))} \,\frac{\alpha}{\beta},
$$
it suffices to check that $\displaystyle \frac{(1-\rho/(\alpha+\zeta))}{(1-\rho/(\beta'+\zeta'))} > 1.$ But this amounts to see that
$$
\beta' + \zeta' < \alpha + \zeta,
$$
i.\,e.
$$
\beta + (1 - \frac\rho{\alpha+\zeta})\,\zeta \,<\, \alpha + \zeta,
$$
which holds since $\beta\le\alpha$ and $\zeta > 0$ (because of \eqref{eq:pos}).
\end{proof}

\medskip\noindent
\textit{Remark 4.1.}
The preceding computations show that if $\alpha/\beta = 1$ and seat $s+1$ is given to $A$, \emph{then} $\alpha'/\beta'= (1-\rho/(\alpha+\zeta))\,\alpha/\beta<1$ and seat $s+2$ is given necessarily to $B$ 
(whenever $s+2\le n$).
Besides, $\alpha''/\beta'' > \alpha/\beta = 1$ and seat $s+3$ is given necessarily to $A$ (whenever $s+3\le n$). And, continuing with the same argument, $\alpha'''/\beta''' < 1,$ $\alpha''''/\beta'''' > 1$ and so on. So in the case $\zeta\ini>0$ a tie of the form $\alpha/\beta = 1$ can happen only once.

\begin{lemma}
\label{st:lemma3}
Assume that $\zeta\ini=0$ and that seat $s+1$ is the matter of a tie, i.\,e.~$\alpha=\beta$. If that seat is allocated to $A$ then seat $s+2$ is allocated to $B$ (whenever $s+2\le n$) and seat $s+3$ is again the matter of a tie (whenever $s+3\le n$). Analogously happens for $A$ interchanged with $B.$ 
\end{lemma}

\begin{proof}
It suffices to check that $\alpha'/\beta'= (\alpha-\rho)/\beta<1$ and that $\alpha''/\beta''= (\alpha-\rho)/(\beta-\rho)=1.$
The only problem could be having $\alpha=\rho$ and $\beta=\rho,$ but this means that each of the previous seats has been given in exchange of a whole quota $\rho$; since $(n+1)\,\rho=1$ and $\alpha+\beta=2\rho,$ it follows that $n-1$ seats have already been given, so that seat $s+1$ is already the last one. 
\end{proof}

Both in the case $\zeta\ini>0$ and in the case $\zeta\ini=0$ the preceding facts have the following consequence:

\begin{corollary}
\label{st:corollary}
Let $k$ be the largest non-negative integer such that the first $k$ seats are allocated all of them to the same party without ties (if $\alpha=\beta,$ then $k=0$). The remaining $n-k$ seats are then divided between the two parties either equally or with a difference of one seat.
\end{corollary}


Our aim is now to compute the limits of $n_A/n$ and $n_B/n$ as $n\rightarrow\infty.$
This computation will rely on estimating the value of the integer $k$
whose definition is contained in the preceding corollary.
Without loss of generality, we can assume that the first seat is allocated to party $A.$


In the remainder of this section we switch back to denoting the values at step $s$ by $\alpha\ara,\,\beta\ara,\,\zeta\ara$,
and the initial ones as $\alpha,\,\beta,\,\zeta$ ($\alpha+\beta+\zeta=1$). The hypothesis that the first seat is allocated to $A$ implies that $\alpha\ge\beta$.

According to its definition,
$k$ is the first integer such that $\alpha[k] \le \beta[k].$
Its value can be computed in the following way.
To start with, the fact that the first $k$ seats are all allocated to $A$ ensures that
\begin{equation}
\label{eq:krho}
\alpha[k]+\zeta[k] = \alpha +\zeta - k\,\rho.
\end{equation}
Each of these seats entails a reduction factor that affects all the ballots that approve of $A,$ either alone or together with $B$. Trom this it follows that
\begin{equation}
\label{eq:proporcio}
\dfrac{\alpha[k]}{\alpha} = \dfrac{\alpha[k] +\zeta[k]}{\alpha+\zeta}.
\end{equation}
Combining \eqref{eq:krho} and \eqref{eq:proporcio} results in
\begin{equation}
\alpha[k] = \left(1-\dfrac{k\,\rho}{\alpha+\zeta}\right)\,\alpha.
\end{equation}
On the other hand, we know that $\beta[k] = \beta.$
These equalities determine the value of the first integer $k$ for which $\alpha[k] \le \beta[k],$ namely
\begin{equation}
k = \left\lceil \dfrac{(\alpha-\beta)(\alpha+\zeta)}{\alpha\,\rho} \right\rceil,
\end{equation}
where $\lceil x\rceil$ means the smallest integer larger than or equal to $x$.
Since $\rho = 1/(\nseg)$, we get
\begin{equation}
\lim_{n\rightarrow\infty}\dfrac{k}{n} = \dfrac{(\alpha-\beta)(\alpha+\zeta)}{\alpha}.
\end{equation}

Finally, it only remains to take into account that $k$ is related to $n_A$ in the following way:
$n_A = k + (n-k)/2$ for $n-k$ even, $n_A = k + (n-k\pm1)/2$ for $n-k$ odd
(the `$+$' sign happens only in the case of ties).
From these facts it follows that
\begin{equation}
\lim_{n\rightarrow\infty}\dfrac{n_A}{n} = \onehalf \left(1+\dfrac{(\alpha-\beta)(\alpha+\zeta)}{\alpha}\right),
\qquad\text{for $\alpha \ge \beta.$}
\end{equation}
The preceding constraint $\alpha \ge \beta,$ corresponds to the assumption that we have made that the first seats are allocated to $A$.
For the contrary case, the analogous formula for $\lim_{n\rightarrow\infty}\,(n_B/n)$ amounts to
\begin{equation}
\lim_{n\rightarrow\infty}\dfrac{n_A}{n} = \onehalf \left(1-\dfrac{(\beta-\alpha)(\beta+\zeta)}{\beta}\right),
\qquad\text{for $\alpha \le \beta.$}
\end{equation}
In particular, both formulas coincide in giving $\lim_{n\rightarrow\infty}\,(n_A/n) = 1/2$ for $\alpha=\beta.$
On the other hand, for $\zeta=0$ both of them become $\lim_{n\rightarrow\infty}\,(n_A/n) = \alpha$
(since $\alpha+\beta=1$).

Figure~\ref{fig:nocantor} shows the way that the limit depends on $\alpha$ for a fixed value of $\zeta$
(in which case $\beta = 1 - \zeta - \alpha$).

\vskip-10mm 

\begin{figure}[H]
\centering
\begin{tikzpicture}[x=90mm,y=90mm,>=stealth]
\draw (0,0) -- (1,0);
\node at (0.85,-4mm) {$\alpha$};
\draw (0,0) -- (0,1);
\draw (1,0) -- (1,1);
\draw (0,1) -- (1,1);
\draw (-1.5mm,0) -- (0,0);
\node at (-4mm,0) {$0$};
\draw (-1.5mm,1) -- (0,1);
\node at (-4mm,1) {$1$};
\draw (0,-1.5mm) -- (0,0);
\node at (0,-4mm) {$0$};
\draw (1,-1.5mm) -- (1,0);
\node at (1,-4mm) {$1\!-\!\zeta$};
\def\z{0.376}
\draw[help lines,domain=0:1,variable=\xx,dashed] plot ({\xx},{\xx});
\draw[domain=0.5:1,variable=\xx] plot ({\xx},{0.5 + ((1-\z)*\xx+\z)*(\xx-0.5)/\xx});
\draw[domain=0:0.5,variable=\xx] plot ({\xx},{0.5 - ((1-\z)*(1-\xx)+\z)*(0.5-\xx)/(1-\xx)});
\end{tikzpicture}%
\caption{Dependence of $\lim_{n\rightarrow\infty}(n_A/n)$ with respect to $\alpha$ for $\zeta = 0.376.$}
\label{fig:nocantor}
\end{figure}
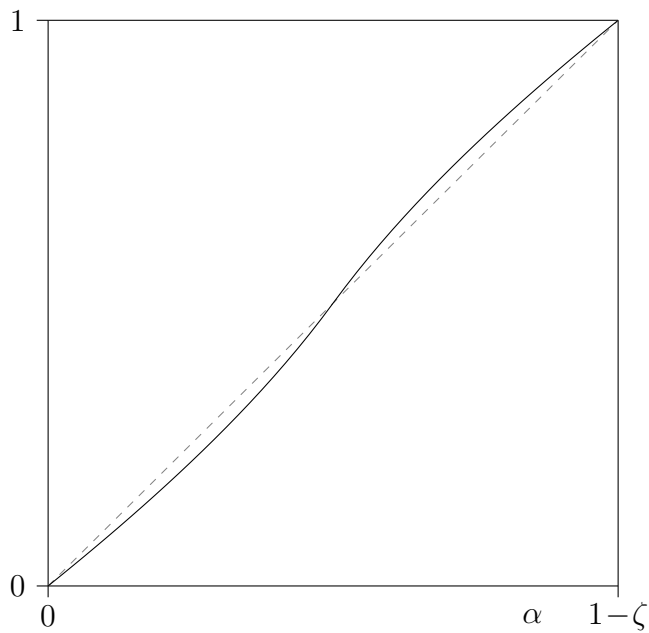


\section{Comparison with other methods}

In this section the above properties of the method of Eneström and Phragmén
will be compared with those of the main alternative methods for parliamentary elections through approval voting.
More specifically, we will consider
both Phragmén's iterative minimax method (\citeyear{phragmen:1894},\,\citeyear{phragmen:1895},\,\citeyear{phragmen:1896},\,\citeyear{phragmen:1899});
and Thiele's~methods (\citeyear{thiele:1895}),
all of which are extensions of D'Hondt's rule from party lists to open lists.

\subsection{Phragmén's minimax method}
\label{ssec:phrseq}

Phragmén's minimax method hinges on considering how to distribute each representative, i.\,e.~elected candidate,
among ``his'' electors, i.\,e.~those who approved him.
In this connection, it aims at minimizing the inequality between electors in what respects the total amount of representation obtained by each of them.
More specifically,
here Phragmén aimed at minimizing the maximal amount of representation obtained by an elector.
In the general case of open lists, such optimization is not easy to compute.
Which is why, instead of it, Phragmén proposed a greedy sequential procedure where
every step looks for an additional representative
that minimizes that maximal representation.

\medskip
This criterion of minimal inequality of representation between electors
is related to the notion of proportional representation.
Although it differs from the interpretation of this notion in terms of a quota,
it turns out that Phragmén's iterative minimax method satisfies also Theorems~\ref{st:droop} and \ref{st:pjr}
with $q = v/(n+1).$
Equivalently said, it enjoys optimal values of the indices $\pi_{\hbox{\scriptsize\sffamily same}}$ and $\pi_{\hbox{\scriptsize\sffamily PJR}}$ of \citealt[Section~7.2]{janson:2018}
(see also \citealt[Sats~13.5.(ii)]{janson:2012}).

\medskip
Concerning monotonicity, Phragmén's iterative minimax method behaves like that of Eneström and Phragmén:
monotonicity for individual candidates holds \cite[Prop.\,7.10]{mora:2015}
and monotonicity for party lists fails \cite[\S\,7.5]{mora:2015}.
House monotonicity is certainly a different matter
since Phragmén's iterative minimax method satisfies it by construction.

\medskip
On the other hand, concerning the asymptotic behaviour in the case of two parties,
Phragmén's iterative minimax method exhibits a singular behaviour where, for instance,
the smooth curve of Figure~\ref{fig:nocantor} is replaced by a ``devil's staircase''
where every rational value is the image of an interval of positive measure
(see \citealt[\S\,7.7]{mora:2015}, and \citealt[\S\,11.4]{janson:2017}).

\subsection{Thiele's methods}
\label{sec:thiele}

\smallskip
Thiele's methods 
aim at maximizing the total amount of satisfaction of the electors.
In this connection, the satisfaction $\sigma$ of an elector is postulated to depend only 
on the number $h$ of elected candidates that had been approved by that elector;
this dependence $\sigma(h)$ is assumed to be non-decreasing with $\sigma(0) = 0$ and $\sigma(1) = 1.$

More particularly, Thiele paid special attention to the case where this dependence has the following form:
\begin{equation}
\sigma(h) = 1 + \dfrac12 + \dots + \dfrac1h,\qquad \text{with $\sigma(0) = 0.$}
\end{equation}
As one can easily see, this function has the property that in the case of party lists 
the criterion of maximizing the total satisfaction leads to D'Hondt's rule.

As in \S\,\ref{ssec:phrseq}, the computational complexity of the general case of open lists
led Thiele to replace the original optimization criterion
by certain greedy sequential versions,
which are known respectively as Thiele's addition method and Thiele's elimination one.

\paragraph{Thiele's addition method.}
\label{ssec:thiele-addition}
In this method one starts with the empty set and
every step looks for an additional representative that results in a maximum increment of satisfaction.

For the sake of comparison with the method of Eneström and Phragmén,
it is worth noticing that Thiele's addition method can also be viewed in terms of a progressive reduction
of the value of each vote each time that it is used to elect a new candidate.
In fact, it amounts to the following reduction scheme: a ballot is reduced to $1/2$ of its value when it is first used to elect one of its candidates. When a second candidate of a ballot is elected, the value of that ballot is reduced to $1/3$ of its initial value, or equivalently, to $2/3$ of its previous value. Similarly, when a third candidate of a ballot is elected, the value of that ballot is reduced to $1/4$ of its initial value, or equivalently, to $3/4$ of its previous value. And so on.

This has some similitude to the method that we have been discussing, especially its ``simple fractions'' variation (\S\,\ref{pg:simple_fractions}).
However, the reduction factors of Thiele's method 
have nothing to do with
the number of ballots that supported the elected candidate, let alone with comparing this number to any prefixed quota. So, that similitude is only superficial.
In the words of \citealt[p.\,4]{phragmen:1906a}, Thiele's addition rule ``is a purely formal generalisation of D'Hondt's rule, and therefore lacks genuine justification.''

In particular, proportionality properties such as Theorems~\ref{st:droop} and \ref{st:pjr} cease to hold when the votes deviate from the case of disjoint party lists.
Consider, for instance, the following example due to \cite{tenow:1912}:
\begin{equation}
\label{eq:tenow-stra}
1\ a,\xsep 9\ a\,b,\xsep 9\ a\,c,\xsep 9\ b,\xsep 9\ c,\xsep 13\ k\,l\,m,
\end{equation}
which will be compared with
\begin{equation}
\label{eq:tenow-orig}
37\ a\,b\,c,\xsep 13\ k\,l\,m.
\end{equation}
Assume that $n=3.$ The quota is $q = 50/4 = 12.5.$ So, the hypotheses of Theorem~\ref{st:droop} are satisfied with $J = \{k,l,m\}$ and $\ell=1.$
If that theorem were true we should have $n_J\ge1.$ This holds in the case of \eqref{eq:tenow-orig}, where $a,\,b$ and $k$ are successively elected (as in D'Hondt's rule).
However in the case of \eqref{eq:tenow-stra}, Thiele's addition method successively elects $a,\,b$ and $c$ 
(whereas the method of Eneström and Phragmén successively elects $a,\,k$ and $b$).
By the way, one can imagine that \eqref{eq:tenow-orig} are sincere votes and \eqref{eq:tenow-stra} is a strategy that allows the \,$abc$\, party to get all three seats.
In the framework of \cite{janson:2018}, the fact that Thiele's addition method does not comply with Theorems~\ref{st:droop} and \ref{st:pjr} translates into its values of $\pi_{\hbox{\scriptsize\sffamily same}}(\ell,n)$ and $\pi_{\hbox{\scriptsize\sffamily PJR}}(\ell,n)$ being larger than the optimal value $\ell/(\nseg)$ (see \citealt[Section~7.4]{janson:2018}).

\medskip
Concerning monotonicity, Thiele's addition method behaves exactly as Phragmén's iterative minimax one:
monotonicity for individual candidates holds \citep[Theorem 14.2]{janson:2016},
monotonicity for party lists fails,
and house monotonicity holds by construction. 


\medskip
Finally, concerning the asymptotic behaviour in the case of two parties,
Thiele's addition method turns out to comply with the ideal behaviour
$\lim_{n\rightarrow\infty}n_A/n = \alpha/(\alpha+\beta) = \alpha/(1-\zeta)$ \citep[Example 12.10]{janson:2017}.


\paragraph{Thiele's elimination method.}
In contrast to the above addition procedure, 
here one starts with the the set of of all candidates
and every step looks for which of them should be removed in order to obtain a minimum decrement of satisfaction.

It turns out that this procedure does satisfy Theorem~\ref{st:droop} but not \ref{st:pjr} \cite[\S\,7.5]{janson:2018}.
In the particular case of \eqref{eq:tenow-stra} it successively eliminates $m,\,l$ and $a$,
thus resulting in the election of $b,\,c$ and $k$
(by the way this is also the result of Thiele's global optimization criterion).
This result agrees with Theorem~\ref{st:droop}, which grants one seat to the set $J=\{k,l,m\}.$
However, it is also true that this set of elected candidates does not include the most voted one, namely $a.$
Such a fact can be viewed as a major flaw of this method
(which view was endorsed by \citealt[p.\,301--302]{phragmen:1899}).
Another flaw of this procedure is that it lacks even the property of monotonicity for individual candidates.



For completeness, we will also mention that computational experiments 
about the asymptotic behaviour in the case of two parties,
show Thiele's elimination method to comply with the ideal behaviour
$\lim_{n\rightarrow\infty}n_A/n = \alpha/(\alpha+\beta) = \alpha/(1-\zeta)$.

Anyway, the elimination procedure is rather inappropriate for the case where the items that are being approved or not are parties.
In fact, in this case that procedure must begin by considering how many candidates are included in a party list,
in spite of the fact that this number should be rather irrelevant.

%
%
%
%

\newcommand\si{\raise2pt\hbox{\small$\surd$}}
\newcommand\no{$\times$}

\subsection{Comparison table}
\label{ssec:comparison}
Table~1 below summarizes the preceding results for the methods that assume approval voting.
In column `Type' we indicate the behaviour in the case of uninominal voting in terms of party lists: `Dr' means largest remainders with the Droop quota, and `D'H' means D'Hondt's rule. Columns `Thm.\,\ref{st:droop}' and `Thm.\,\ref{st:pjr}' indicate whether these theorems are satisfied or not. In the column `Mono', the value `ind'  means that monotonicity holds for individual candidates but not for party lists (see Section~\ref{sec:mono}) whereas `\no' means that monotonicity fails even for individual candidates. The column~`2Lim' refers to the asymptotic behaviour in the case of two parties; here the value `\si' means that the limit is the ideal fraction $\alpha/(\alpha+\beta),$ the value `\no' is motivated by the devil's staircase phenomenon of Phragmén's iterative minimax method, and the value `$\sim$' means a smooth function not very different from the ideal one. Finally, the column `Simpl' tries to categorize the simplicity of the method in two levels: `$\sim$' -- acceptable, and `\no' -- somewhat complex.

\makeatletter
\renewcommand\strut{\vrule\@height14pt\@depth6pt\@width\z@}
\makeatother
\begin{table}[h]
\begin{center}
\begin{small}
\begin{tabular}{ l | c | c | c | c | c | c |}
\strut								& Type & Thm.\,\ref{st:droop} & Thm.\,\ref{st:pjr} & Mono & 2Lim & Simpl \\ \hline
\strut \textsl{Eneström-Phragmén}	& Dr   & \si  & \si  & ind & $\sim$ & $\sim$ \\ \hline
\strut \textsl{Phragmén, maximin}	& D'H  & \si  & \si  & ind & \no & \no \\ \hline
\strut \textsl{Thiele, addition}	& D'H  & \no  & \no  & ind & \si & $\sim$ \\ \hline
\strut \textsl{Thiele, elimination}	& D'H  & \si  & \no  & \no & \si & \no \\ \hline
\end{tabular}
\end{small}
\captionsetup{width=110mm}
\caption{Comparison of different electoral methods based on approval voting.}
\end{center}
\end{table}


\section{An (unsuccessful) attempt at a divisor-like variation}

For disjoint closed lists,
the rules of largest remainders and D'Hondt are related in the following way:
D'Hondt's rule is the result of adjusting the quota so that
exactly $n$ seats are allocated without making use of any remainders.
In the present setting of (possibly intersecting) open lists,
one can try to do the same from the procedure of Eneström and Phragmén,
namely to adjust the quota so that
exactly $n$ seats are allocated 
under the conditions that every seat is allocated in exchange of one quota
and that none of the other candidates achieves a whole quota.
In our notation, and making allowance for ties,
this amounts to the quota $q$ and the number of seats $n$ being related in the following way:
\begin{equation}
\label{eq:penqfin}
w_*\pen \,\ge\, q \,\ge\, w_*\fin,
\end{equation}
where $w_*\ara$ are obtained by the algorithm of \S~1,
and therefore they depend on $q.$

\bigskip\noindent
\textit{Remark~6.1}.
The preceding inequalities are analogous to the following ones for D'Hondt's rule (see \citealt[\S4.6]{pukelsheim:2017}):
\begin{equation}
\min_i v_i/n_i \,\ge\, q \,\ge\, \max_i v_i/(n_i+1).
\end{equation}



Unfortunately,
such a plan is hindered by several difficulties.
To begin with, for a given $n$ one can have different values of $q$ 
that satisfy \eqref{eq:penqfin} but lead to different allocations of the $n$~seats.
In order to deal with this difficulty, one can think of further specifying $q,$
for instance, by requiring it to be as large as possible.
However, such a condition is not easy to compute.
In fact, it might seem that it amounts to solving the equation $w_*\pen=q$
(where the left-hand side depends on $q$)
but sometimes this equation has no solution at all.

These difficulties occur, for instance, in the following example,
where $A,\,B$ and $C$ are three party lists:
\begin{equation}
\label{eq:exempleFinal}
7\ A,\xsep 10\ B,\xsep 5\ A\,B,\xsep 17\ C,\xsep 13\ A\,C,\xsep 4\ B\,C.
\end{equation}
Figure~2 shows the dependence of $w_*\pen$ as a function of $q$ for $n=2...9$.
As one can see, for $n=2,\,3$ the corresponding curves do not reach the diagonal $w_*\pen=q$.
On the other hand, for $n=6$ the corresponding curve exhibits a discontinuity
that is associated with the allocation changing from $4\,C, \,1\,A, \,1\,B$
to $3\,C, \,2\,A, \,1\,B.$


\vskip3mm
{\leftskip-20mm
\hfil\includegraphics[scale=.8]{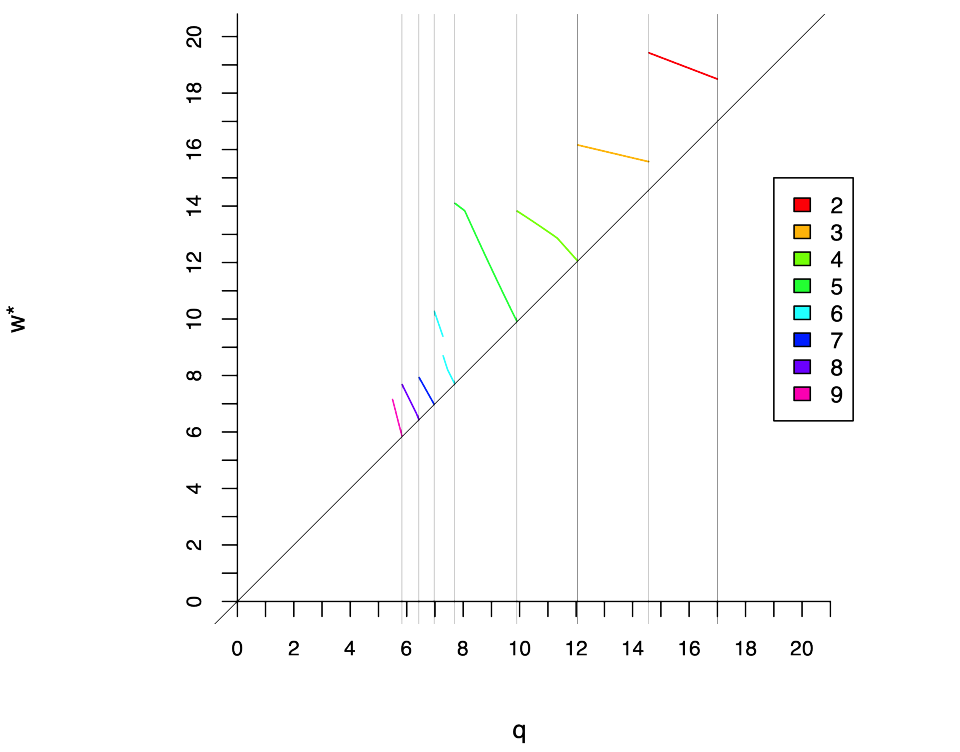}\hfil\par}
\vskip-1mm
\centerline{Figure~2. $w_*\pen$ as a function of $q$ for example~\eqref{eq:exempleFinal}}


\section{Providing for substitutes}
\label{substituts}

In practice, one must envisage the possibility that one of the elected candidates ceases to be available at some point and has to be replaced.
In order to keep the representativeness of the set of elected candidates, the ideal would be for the substitute to represent the same voters as the candidate who has ceased to be available. For party lists, the solution is obvious: just take another candidate from the same party list as a substitute. This applies even when the elector is allowed to approve more than one list, provided that the party list of the missing candidate still contains some unelected one (otherwise, the problem will have to be paraphrased in terms of individual candidates).

\smallskip
In the general case the solution is not so easy.
For the method that has been exposed in this article, a natural way to proceed is the following
(in the same spirit as \citealt[p.\,194--195]{droop:1881}):\,
1.~The votes and vote fractions in exchange for which the missing candidate had been elected are restored 
to the votes and vote fractions remaining after the $n$~seats have been allocated.
2.~ Having done this, the procedure for assigning a new seat is executed once again, with the only caveat that the missing candidate is not considered eligible.

\smallskip
Let us see how we would do it in example~1.1 of p.\,\pageref{exemple11}. Suppose candidate~$a$ is no longer available. This candidate had been elected in the first place, in exchange for 25 votes, specifically the $v_k[0]\!-\!v_k[1]$ that are specified below, which are obtained from the 43~votes where $a$ appears by applying the factor $25/43$:
\bgroup
\advance\abovedisplayskip by-4pt
\advance\belowdisplayskip by-4pt
$$
\hbox to\hsize{\hfill$12,209\ a\,b\,x,\xsep 11,628\ a\,b\,e\,f,\xsep 1,163\ a\,e\,u.$\hfill}
$$
To elect the substitute, we restore these votes to those that were left at the end (p.\,\pageref{romanents}); the result is the following values of $v_k[3]\!+\!v_k[0]\!-\!v_k[1]$:
$$
\renewcommand\xsep{\hskip.5em}
\hbox to\hsize{\hfill\scalebox{0.8}{$12,209\ a\,b\,x,\xsep 20\ a\,b\,e\,f,\xsep 4,535\ e\,f\,u\,v,\xsep 3,103\ u\,v,\xsep 0\ x\,y,\xsep 15\ z,\xsep 1,363\ a\,e \,u.$}\hfill}
$$
The support for each candidate is now as follows, where brackets indicate ineligible candidates (already elected or unavailable):
$$
\renewcommand\xsep{\hskip.5em}
\hbox to\hsize{\hfill\scalebox{0.8}{$(a\ 33,572),\xsep b\ 32,209,\xsep e\ 25,898,\xsep f\ 24,535,\xsep (u\ 9),\xsep v\ 7,637,\xsep (x\ 12,209),\xsep y\ 0,\xsep z\ 15.$ }\hfill}
$$
\egroup
Of the eligible candidates, the one with the most votes is $b,$ who is therefore elected to replace $a.$ This is no surprise, since these two candidates appear together most often. Similarly happens with $u$ and $v,$ so that in the event of $u$ ceasing to be available, it is replaced by $v,$ as obtained by considering the values of $v_k[3]\!+\!v_k[1]\!-\!v_k[2].$ On the other hand, $x,$ the candidate elected in third place, appears sometimes accompanied by $y$ and other times by $b$ (and $a$); if we apply the proposed procedure, we must consider $v_k[3]\!+\!v_k[2]\!-\!v_k[3] = v_k[2],$ whose values are collected in p.\,{remaining}, where it is seen that they result in the choice of $b$ as a substitute for $x.$

\bigskip\noindent
\textit{Remark~7.1.} The proposed procedure is not equivalent to running the entire algorithm from the beginning after having deleted the missing candidate. Indeed, doing so could lead to variations in the subsequently elected candidates.

\bigskip\noindent
\textit{Remark~7.2.} If the candidates who cease to be available are more than one, then dealing with them successively in one order or another can produce different results, and dealing with them all at once can aslso lead to different results. Therefore, in practice it will be necessary for the regulations to specify one of these different alternatives.

\bigskip\noindent
\textit{Remark~7.3.} \citeauthor{phragmen:1906a} (\citeyear{phragmen:1906a}, \citeyear{phragmen:1906b}, \citeyear{phragmen:1906c}) proposes a different procedure where each voter can specify his ``substitutes'', which would be understood as additional candidates that this voter still admits as epresentatives of him, but in the second instance. However, our proposal above is conceptually clearer and has a more general applicability (for example if the candidate who becomes unavailable is already a substitute).

\bigskip\noindent
\textit{Remark~7.4.} The preceding idea does not extend easily to the sequential minimax method of \citeauthor{phragmen:1894} (\citeyear{phragmen:1894}, \citeyear{phragmen:1895}, \citeyear{phragmen:1896}, \citeyear{phragmen:1899}) for which no procedure seems to have been proposed for the provision of substitutes.


\end{document}